\documentclass[reqno]{amsart}%
\usepackage{amsfonts}
\usepackage{verbatim}
\usepackage{xcolor}
\usepackage{amsmath}
\usepackage{amssymb}
\usepackage{graphicx}
\usepackage{accents}%
\providecommand{\U}[1]{\protect\rule{.1in}{.1in}}
\newtheorem{theorem}{Theorem}[section]

\newtheorem{proposition}[theorem]{Proposition}
\newtheorem{remark}[theorem]{Remark}
\newtheorem{lemma}[theorem]{Lemma}

\numberwithin{equation}{section}
\begin{document}
\title[KdV equation]{On the inverse scattering transform for the KdV equation with summable initial data}
\author{Alexei Rybkin}
\address{Department of Mathematics and Statistics, University of Alaska Fairbanks, PO
Box 756660, Fairbanks, AK 99775}
\email{arybkin@alaska.edu}
\thanks{The author is supported in part by the NSF grant DMS-2307774.}
\date{April, 2026}

\begin{abstract}
We consider the Cauchy problem for the Korteweg--de Vries equation with real
initial data $q\in L^{1}\cap L^{2}$ supported on $(0,\infty)$. Using the left
reflection coefficient and Hankel operators on the Hardy space $H^{2}%
(\mathbb{C}^{+})$, we derive a trace-type representation for the corresponding
solution. The proof is based on approximation by compactly supported
potentials, uniform convergence of the associated reflection coefficients away
from the origin, and continuity properties of the resulting Hankel operators.
This yields a rigorous inverse scattering construction for a class of summable
half-line supported initial data beyond the standard short-range setting.

\end{abstract}
\maketitle
\dedicatory{This work is dedicated to the memory of Vladimir A. Marchenko, whose
fundamental contributions to inverse spectral theory and the inverse
scattering transform have profoundly shaped the field. His ideas continue to
influence modern developments in integrable systems.}

\section{Introduction}

We are concerned with the inverse scattering transform (IST)\ for the Cauchy
problem for the Korteweg-de Vries equation%
\begin{equation}%
\begin{cases}
\partial_{t}q-6q\partial_{x}q+\partial_{x}^{3}q=0,\ \ \ x\in\mathbb{R}%
,t\geq0\\
q(x,0)=q(x).
\end{cases}
\label{KdV}%
\end{equation}
This is of course a classical well-known problem solved in the seminal 1967
Gardner-Greene-Kruskal-Miura paper \cite{GGKM67} who solved it for Schwarz
initial data $q$. Conceptually, the IST is similar to the Fourier method but
is based on the direct/inverse scattering\footnote{If $q$ is periodic then in
place of the IST one needs to use the inverse spectral transform which is also
well-developed.} theory for the 1D Schr\"{o}dinger operator $\mathbb{L}%
_{q}=-\partial_{x}^{2}+q(x)$. More specifically, solution to (\ref{KdV})\ is
given in terms of the scattering data $S_{q}$ (see below) via a three step
procedure. The crucial use of the inverse scattering problem imposes a severe
limitation on $q$ - it has to satisfy
\begin{equation}
\int_{\mathbb{R}}\left(  1+x\right)  \left\vert q\left(  x\right)  \right\vert
\mathrm{d}x<\infty\text{ (short-range).}\label{short range}%
\end{equation}
Any relaxation of (\ref{short range}) meets serious complications. In
\cite{ADM81} two potentials $q_{\pm}$, supported on $\left(  \pm
\infty,0\right)  $, and behaving like $q\left(  x\right)  =O\left(  \left\vert
x\right\vert ^{-2}\right)  $, $x\rightarrow\pm\infty$, respectively, are
explicitly constructed , that share the same reflection coefficients. I.e. the
scattering data $S_{q}$ no longer determine $q$ uniquely and scattering data
are no longer data. This of course means that the IST method, as we know it,
break down. The uniqueness in \cite{ADM81} can be easily restore by specifying
the support of $q$ \ Moreover the situation remains manageable if $q$ is
assumed to be short-range only at $+\infty$. Then there exists a suitable set
of scattering data from the right that allows to push the IST method to
essentially any $q$ such that the spectrum of $\mathbb{L}_{q}$ is bounded
below. Such $q$ is commonly referred as to step-like or step-type (see
\cite{GruRybBLMS20} and the literature cited therein). The decay of type
(\ref{short range}) is still assumed at $+\infty$ and this assumption is
crucial for the technics of \cite{GruRybBLMS20} to work.

The main goal of our contribution is to address this problem when
(\ref{short range}) is violated at $+\infty$. Of course, due to unidirectional
nature of the KdV equation, the simple trick of considering $q\left(
-x,-t\right)  $ (which clearly solves KdV) does not work. Indeed, while
$+\infty$ and $-\infty$ do switch, but the technics of \cite{GruRybBLMS20}
dramatically fail for $t<0$. The main reason is that while scattering matrix
is well-defined for any $L^{1}$ potential but continuity is lost at the zero
energy. In other words, zero energy (momentum) is a spectral singularity which
is notoriously hard to control. By the same token, we don't know if it is
possible at all to associate with a zero spectral singularity a norming
constant that would restore uniqueness. In this paper, we develop a systematic
framework in which detour this circumstance by imposing a condition that $q\in
L^{1}$ are supported on $\left(  0,\infty\right)  $. This does not remove a
spectral singularity at $0$ but use the techniques of Hankel operators
developed in \cite{GruRybSIMA15}, \cite{GruRybBLMS20}, and \cite{RybNON24} are
robust enough to circumvent the issues. We note that in the context of
Wigner-von Neumann initial data certain data behaving like $q\left(  x\right)
=0(1/x)$ as $x\rightarrow+\infty$ is recently considered in
\cite{GruRybNON2022},\cite{RybCMP23}. We believe that putting $L^{1}$
potentials into IST setting is approached on the rigorous bases for the first
time here. We discuss challenges below in detail once we set up our framework.

The structure of the paper is as follows. In Section \ref{notation} we fix
notation. Section \ref{HO} reviews the Hardy-space and Hankel-operator tools
used throughout the paper. In Section \ref{Sect Overview} we briefly recall
the relevant facts from one-dimensional scattering theory and the associated
IST for compactly supported potentials, with particular emphasis on the left
reflection coefficient and on the approximation by truncated potentials.
Section \ref{sect of kdv solutions} contains the proof of the main trace
formula for KdV solutions corresponding to real $L^{1}\cap L^{2}$ initial data
supported on $(0,\infty)$. The appendix collects several auxiliary facts used
in the limiting arguments.

\section{Notations\label{notation}}

Our notations are quite standard:

\begin{itemize}
\item $\mathrm{1}_{S}$ is the characteristic function of a (measurable) set
$S$.

\item $L^{p}\left(  \mathbb{R}\right)  =L^{p}$, $1\leq p\leq\infty$;
$\left\Vert \cdot\right\Vert _{p}$ is the $L^{p}$ norm.

\item We write $x\lesssim_{a}y$ if $x,y\geq0$ and $x\leq C\left(  a\right)  y$
with a positive $C$ dependent on $a$. We drop $a$ if $C$ is a universal constant.

\item As always for a Hilbert space operator $T$ we $T\geq0$ if $\left\langle
Tf,f\right\rangle \geq0$ for all $f$ (a positive operator).
\end{itemize}

\section{Hardy spaces and Hankel operators\label{HO}}

To fix our notation we review some basics of Hardy spaces and Hankel operators
following \cite{Nik2002},\cite{Peller2003}

A function $f$ analytic in $\mathbb{C}^{\pm}=\left\{  z\in\mathbb{C}%
:\pm\operatorname{Im}z>0\right\}  $ is in the Hardy space $H_{\pm}^{p}$ for
some $0<p\leq\infty$ if
\[
\Vert f\Vert_{H_{\pm}^{p}}^{p}\overset{\operatorname*{def}}{=}\sup_{y>0}\Vert
f(\cdot\pm\mathrm{i}y)\Vert_{p}<\infty.
\]
We set $H^{p}=H_{+}^{p}.$ Every $f\in H_{\pm}^{p}$ has non-tangential boundary
values $f\left(  x\pm\mathrm{i}0\right)  $ for almost every (a.e.)
$x\in\mathbb{R}$ and%
\begin{equation}
\Vert f\Vert_{H_{\pm}^{p}}=\Vert f\left(  \cdot\pm\mathrm{i}0\right)
\Vert_{L^{p}}=\left\Vert f\right\Vert _{L^{p}}. \label{H^p norm}%
\end{equation}
Recall that $H_{\pm}^{2}$ are Hilbert spaces with the inner product induced
from $L^{2}$. It is well-known that $L^{2}=H^{2}\oplus H_{-}^{2},$ the
orthogonal (Riesz) projection $\mathbb{P}_{\pm}$ onto $H_{\pm}^{2}$ being
given by%
\begin{equation}
(\mathbb{P}_{\pm}f)(x)=\pm\frac{1}{2\pi\mathrm{i}}\lim_{\varepsilon
\rightarrow0+}\int_{\mathbb{R}}\frac{f(s)\mathrm{d}s}{s-(x\pm\mathrm{i}%
\varepsilon)}=:\pm\frac{1}{2\pi\mathrm{i}}\int_{\mathbb{R}}\frac
{f(s)\ \mathrm{d}s}{s-(x\pm\mathrm{i}0)}. \label{proj}%
\end{equation}
The class $H^{\infty}$ is the algebra of functions uniformly bounded in
$\mathbb{C}^{+}$.

Define the Hankel operator on $H^{2}$. Let $(\mathbb{J}f)(x)=f(-x)$ be the
operator of reflection. Given $\varphi\in L^{\infty}$ the operator
$\mathbb{H}(\varphi):H^{2}\rightarrow H^{2}$ given by the formula%
\begin{equation}
\mathbb{H}(\varphi)f=\mathbb{JP}_{-}\varphi f,\ \ \ f\in H^{2}, \label{Hankel}%
\end{equation}
is called the Hankel\emph{ }operator with symbol $\varphi$.

Clearly $\left\Vert \mathbb{H}(\varphi)\right\Vert \leq\left\Vert
\varphi\right\Vert _{L^{\infty}}$, $\mathbb{H}(\varphi)$ is self-adjoint if
$(\mathbb{J}\varphi)(x)=\overline{\varphi\left(  x\right)  }$ (this is always
our case), $\mathbb{H}(\varphi)=0$ if $\varphi$ is a constant, and if
$\varphi\in L^{2}\cap L^{\infty}$%
\begin{equation}
\mathbb{H}(\varphi)=\mathbb{H}(\mathbb{P}_{-}\varphi)\text{.}
\label{regularized Hankel}%
\end{equation}
We will also need the Nehari theorem saying that%
\[
\left\Vert \mathbb{H}(\varphi)\right\Vert \lesssim\left\Vert \mathbb{P}%
_{-}\varphi\right\Vert _{BMO},
\]
where
\[
\left\Vert f\right\Vert _{BMO}=\sup_{I\in\mathbb{R}}\frac{1}{\left\vert
I\right\vert }\int_{I}\left\vert f-f_{I}\right\vert ,\ \ \ f_{I}:=\frac
{1}{\left\vert I\right\vert }\int_{I}f
\]
is the BMO norm (boundary mean oscillation) of $f$.

Finally, the famous Hartman theorem says that $\mathbb{H}(\varphi)$ is a
compact operator if and only if $\varphi\in H^{\infty}+C$ (aka the Sarason
algebra), where $C$ is the set of continuous functions on the one point
compactification of $\mathbb{R}$.

The relevance of the Hankel operator in the IST setting is obvious as the
classical Marchenko operator is a Hankel operator on $L^{2}\left(
0,\infty\right)  $. In our approach, however, Hankel operators on $H^{2}$ are
more convenient.

Finally we emphasize that the use of Hankel-operator methods in the study of
integrable systems has recently gained momentum. In addition to their role in
the classical Marchenko formalism, such methods now appear in
direct-linearization and Grassmannian-flow approaches to a number of local,
nonlocal, and noncommutative integrable equations; see, for example,
\cite{Blower23}, \cite{Doikouetal21}, \cite{DoikouMalhamStylianidisWiese23},
\cite{Malham22}, \cite{BlowerMalham25} and the references cited therein. See
also \cite{GerardPushnitski2023} that develops a spectral theory for unbounded
Hankel operators and shows how their evolution encodes the dynamics of the
cubic Szeg\H{o} equation.

\section{Overview of 1D scattering theory and associated
IST\label{Sect Overview}}

Through this section, unless otherwise stated, we assume that $q$ is supported
on $\left(  0,b\right)  $ with finite $b>0$ and $q\in L^{2}$. Apparently such
$q$'s are subject to (\ref{short range}) and therefore the short-range
scattering theory fully applies. Unless otherwise stated all facts are taken
from \cite{MarchBook2011}. Associate with $q$ the full line Schr\"{o}dinger
operator $\mathbb{L}_{q}=-\partial_{x}^{2}+q(x)$. As is well-known,
$\mathbb{L}_{q}$ is self-adjoint on $L^{2}$ and its spectrum consists of $N$
simple negative eigenvalues $\{-\kappa_{n}^{2}:1\leq n\leq N\},$ called bound
states ($N=0$ if there are no bound states), and two fold absolutely
continuous component filling $\left(  0,\infty\right)  $. There is no singular
continuous spectrum. Two linearly independent (generalized) eigenfunctions of
the a.c. spectrum $\psi_{\pm}(x,k),\;k\in\mathbb{R}$, can be chosen to
satisfy
\begin{equation}
\psi_{+}(x,k)=\mathrm{e}^{\mathrm{i}kx},x>a,\;\psi_{-}(x,k)=\mathrm{e}%
^{-\mathrm{i}kx},x<0.\label{Jost sltns}%
\end{equation}
The function $\psi_{\pm}$, referred to as right/left Jost solution of the
Schr\"{o}dinger equation%
\begin{equation}
\mathbb{L}_{q}\psi=k^{2}\psi,\label{eq6.3}%
\end{equation}
is analytic for $\operatorname{Im}k>0$. For real $k$, introduce ($W\left(
f,g\right)  =fg^{\prime}-f^{\prime}g$)%
\begin{align}
T\left(  k\right)   &  =\frac{2\mathrm{i}k}{W(\psi_{-},\psi_{+})},\ R\left(
k\right)  =\frac{W(\overline{\psi_{+}},\psi_{-})}{W(\psi_{-},\psi_{+}%
)}=T\left(  k\right)  \frac{W(\overline{\psi_{+}},\psi_{-})}{2\mathrm{i}%
k},\label{T}\\
L\left(  k\right)   &  =\frac{W(\psi_{+},\overline{\psi_{-}})}{W(\psi_{-}%
,\psi_{+})}=T\left(  k\right)  \frac{W(\psi_{+},\overline{\psi_{-}}%
)}{2\mathrm{i}k},\ \label{RL}%
\end{align}
where $T,R,L$, called the transmission, right and left reflection coefficients
respectively, which are independent of $x\;$($\partial_{x}\;$is missing in
$\mathbb{L}_{q}$). It is well-known that $T\left(  k\right)  $ is meromorphic
for $\operatorname{Im}k>0$ with finitely many simple poles at $\left(
\mathrm{i}\kappa_{n}\right)  $\footnote{Recall $-\kappa_{n}^{2}\in
\operatorname{Spec}(\mathbb{L}_{q}),\;n=1,2,\ldots,N.$} and continuous for
$\operatorname{Im}k=0$. Generically, $T\left(  0\right)  =0$. Due to
(\ref{Jost sltns}) one concludes \cite{Deift79} that $R$ and $L$ are also
meromorphic in $\mathbb{C}^{+}$ away from $\left(  \mathrm{i}\kappa
_{j}\right)  $. Moreover, and it is very important, for the corresponding
residue\footnote{We will only need $L$.} we have \cite{AK01}%
\begin{equation}
\operatorname{Res}\left(  L,\mathrm{i}\kappa_{n}\right)  =\mathrm{i}%
c_{n},\label{residues}%
\end{equation}
where $c_{n}=\left\Vert \psi_{-}(\cdot,\mathrm{i}\kappa_{n})\right\Vert
_{2}^{-2}$ are (left) norming constants. As is well-known (see, e.g.
\cite{MarchBook2011}), the triple $\{L,(\kappa_{j},c_{,j})\},$ called the
(left) scattering data for $\mathbb{L}_{q}$, determine $q$ uniquely.

What is important to us is that in the case of finite support $L$ alone
determines $q$. This is particularly easy to see from (\ref{T}), (\ref{RL})
that
\begin{equation}
L\left(  k\right)  =\frac{\mathrm{i}k-m_{+}\left(  k^{2}\right)  }%
{\mathrm{i}k+m_{+}\left(  k^{2}\right)  },\label{L}%
\end{equation}
where $m_{+}$ is the (principal) Titchmarsh-Weyl m-function corresponding to
$\left(  0,\infty\right)  $, which is due to the famous Borg-Marchenko
uniqueness result, determines $q$ uniquely. Note that $m_{+}$ is well-defined
for any $q$ subject to the limit point case at $+\infty$, and thus so is $L$
defined by (\ref{L}). The representation (\ref{L}) also implies
\cite{GruRybSIMA15} that $L\left(  k\right)  $ can be analytically continued
into $\operatorname{Im}k>0$ and the sequence of $\left(  c_{n}\right)  $ is
summable, phenomena of the restricted support of $q$. Note that for general
$L^{1}$ potentials ($q$ supported on $\mathbb{R}$) there is no restrictions on
$\left(  c_{n}\right)  $.

Since many of our proofs below are based on limiting arguments we need to
understand in what sense the scattering data of $q_{b}:=q\mathrm{1}_{\left(
0,b\right)  }$ converge to that of $q$ as $b\rightarrow+\infty$. In general we
can only claim that%
\begin{equation}
L_{b}\left(  k\right)  \rightarrow L\left(  k\right)  ,\ \ \ b\rightarrow
+\infty, \label{unif conv}%
\end{equation}
uniformly on compacts of $\mathbb{C}^{+}$, which only guarantees weak
convergence on the real line. (See \cite{RemlingCMP15} for precise
statements). In our case the convergence can be upgraded to uniform away from
$0$ and this is the principal fact.

Note that for $L^{1}$ potentials $L\left(  k\right)  $ is continuous away from
$0$. What happens at zero is yet to be understood. The paper \cite{Yafaev81}
suggests that it is a very difficult question.

Turn now to the IST for (\ref{KdV}). The starting point of our approach is the
following statement\footnote{It is proven under less restrictive conditions.}
from \cite{RybNON24}.

\begin{proposition}
[Trace formula for KdV solution]\label{Prop on trace formula}Suppose $q\in
L^{2}$ and compactly supported. Let
\[
\varphi_{x,t}\left(  k\right)  =%
{\displaystyle\sum_{n=1}^{N}}
\frac{-\mathrm{i}c_{n}\xi_{x,t}\left(  \mathrm{i}\kappa_{n}\right)  ^{-1}%
}{k-\mathrm{i}\kappa_{n}}+\xi_{x,t}\left(  k\right)  ^{-1}L\left(  k\right)  ,
\]
where%
\[
\xi_{x,t}\left(  k\right)  :=\exp\mathrm{i}\left(  8k^{3}t+2kx\right)  .
\]
Then for the solution to (\ref{KdV}) (necessarily unique) we have
\begin{align*}
&  q\left(  x,t\right)  \\
&  =-\partial_{x}\left\{  2%
{\displaystyle\sum_{n=1}^{N}}
c_{n}\xi_{x,t}\left(  \mathrm{i}\kappa_{n}\right)  ^{-1}m\left(
\mathrm{i}\kappa_{n},x,t\right)  +\int_{\mathbb{R}}\xi_{x,t}\left(  k\right)
^{-1}L\left(  k\right)  m\left(  k,x,t\right)  \frac{\mathrm{d}k}{\pi
}\right\}  ,
\end{align*}
where%
\[
m\left(  \cdot,x,t\right)  =1-\left[  I+\mathbb{H}(\varphi_{x,t})\right]
^{-1}\mathbb{H}(\varphi_{b,x,t})1.
\]

\end{proposition}

\section{Trace formula and KdV solutions\label{sect of kdv solutions}}

In this section we state and prove our main result.

\begin{theorem}
\label{main thm}Let $q\left(  x\right)  \in L^{1}\cap L^{2}$ be real and
supported on $\left(  0,\infty\right)  $ and $L\left(  k\right)  $ the
corresponding left reflection coefficient. Then the solution\footnote{The
famous Bourgain theorem \cite{Bourgain93} guaranties well-posedness in $L^{2}%
$.} $q\left(  x,t\right)  $ to the Cauchy problem for the KdV equation
(\ref{KdV}) with initial data $q\left(  x\right)  $ is given by the trace
formula%
\begin{equation}
q\left(  x,t\right)  =-\partial_{x}\int_{\Gamma}\xi_{x,t}\left(  k\right)
^{-1}L\left(  k\right)  m\left(  k,x,t\right)  \frac{\mathrm{d}k}{\pi
},\ \ \ t>0. \label{KdV solution}%
\end{equation}
Here $\Gamma=$ $\left(  -\infty,a\right)  \cup$ $\left[  -a,-a+\mathrm{i}%
a,a+\mathrm{i}a,a\right]  \cup\left(  a,+\infty\right)  $ is the contour
oriented from left to right and $a$ is large enough for $\Gamma$ to be above
all poles of $L\left(  k\right)  $;%
\[
\xi_{x,t}\left(  k\right)  :=\exp\mathrm{i}\left(  8k^{3}t+2kx\right)  ;
\]
and%
\[
m\left(  \cdot,x,t\right)  =1-\left[  I+\mathbb{H}(\Phi_{x,t})\right]
^{-1}\mathbb{J}\Phi_{x,t},
\]
where $(\mathbb{J}f)(x)=f(-x)$ and $\mathbb{H}(\Phi_{x,t})$ is the Hankel
operator (\ref{Hankel}) with symbol%
\[
\Phi_{x,t}\left(  k\right)  =-\int_{\Gamma}\frac{\xi_{x,t}\left(
\lambda\right)  ^{-1}L\left(  \lambda\right)  }{\lambda-\left(  k-\mathrm{i}%
0\right)  }\frac{\mathrm{d}\lambda}{2\pi\mathrm{i}}%
\]

\end{theorem}

\begin{proof}
Our approach is based on limiting arguments. Let $q_{b}\left(  x\right)  $ be
the restriction of $q\left(  x\right)  $ to $\left(  0,b\right)  $. The full
machinery of the classical inverse scattering transform then clearly applies.
In particular, by Proposition \ref{Prop on trace formula}%
\begin{align}
&  q_{b}\left(  x,t\right)  \label{q_b}\\
&  =-\partial_{x}\left\{  2%
{\displaystyle\sum_{n=1}^{N_{b}}}
c_{b,n}\xi_{x,t}\left(  \mathrm{i}\kappa_{b,n}\right)  ^{-1}m_{b}\left(
\mathrm{i}\kappa_{n},x,t\right)  +\frac{1}{\pi}\int_{\mathbb{R}}\xi
_{x,t}\left(  k\right)  ^{-1}L_{b}\left(  k\right)  m_{b}\left(  k,x,t\right)
\mathrm{d}k\right\}  ,\nonumber
\end{align}
where all quantities with subscript $b$ correspond to those of $q_{b}$. For
$m_{b}$ we have%
\begin{equation}
m_{b}\left(  \cdot,x,t\right)  =1-\left[  I+\mathbb{H}(\varphi_{b,x,t}%
)\right]  ^{-1}\mathbb{H}(\varphi_{b,x,t})1,\label{m_b}%
\end{equation}
where%
\[
\varphi_{b,x,t}\left(  k\right)  =%
{\displaystyle\sum_{n=1}^{N_{b}}}
\frac{-\mathrm{i}c_{b,n}\xi_{x,t}\left(  \mathrm{i}\kappa_{b,n}\right)  ^{-1}%
}{k-\mathrm{i}\kappa_{b,n}}+\xi_{x,t}\left(  k\right)  ^{-1}L_{b}\left(
k\right)  .
\]
Note now that since $L_{b}\left(  k\right)  $ is meromorphic in
$\operatorname{Im}k>0$ with simple poles at \textrm{$i$}$\kappa_{b,n}$, the
integrand in (\ref{q_b}) $\xi_{x,t}\left(  k\right)  ^{-1}L_{b}\left(
k\right)  m_{b}\left(  k,x,t\right)  $ \ is also meromorphic in
$\operatorname{Im}k>0$, with the same poles. For residues we have%
\[
\operatorname*{Res}\left\{  \xi_{x,t}\left(  k\right)  ^{-1}L_{b}\left(
k\right)  m_{b}\left(  k,x,t\right)  ,\mathrm{i}\kappa_{b,n}\right\}
=\mathrm{i}c_{b,n}\xi_{x,t}\left(  \mathrm{i}\kappa_{b,n}\right)  ^{-1}%
m_{b}\left(  \mathrm{i}\kappa_{n},x,t\right)  ,
\]
where we have used the important fact that $\operatorname*{Res}\left\{
L_{b}\left(  k\right)  ,\mathrm{i}\kappa_{b,n}\right\}  =$\textrm{$i$}%
$c_{b,n}$. By deforming the contour $\mathbb{R}$ to $\Gamma$ we obviously
eliminate the sum in (\ref{q_b}) and therefore get%
\begin{equation}
q_{b}\left(  x,t\right)  =-\partial_{x}\int_{\Gamma}\xi_{x,t}\left(  k\right)
^{-1}L_{b}\left(  k\right)  m_{b}\left(  k,x,t\right)  \frac{\mathrm{d}k}{\pi
},\label{q_b 1}%
\end{equation}
which is our starting point for taking $b$ to infinity. Turn to $m_{b}$ given
by (\ref{m_b}). Due to (\ref{regularized Hankel}),
\[
\mathbb{H}(\varphi_{b,x,t})=\mathbb{H}(\Phi_{b,x,t}),
\]
where%
\begin{align}
\Phi_{b,x,t}\left(  k\right)   &  =\mathbb{P}_{-}\varphi_{b,x,t}%
=\mathbb{P}_{-}\left[
{\displaystyle\sum_{n=1}^{N_{b}}}
\frac{-\mathrm{i}c_{b,n}\xi_{x,t}\left(  \mathrm{i}\kappa_{b,n}\right)  ^{-1}%
}{k-\mathrm{i}\kappa_{b,n}}\right]  +\mathbb{P}_{-}\left[  \xi_{x,t}\left(
k\right)  ^{-1}L_{b}\left(  k\right)  \right]  \label{pre big fi}\\
&  =%
{\displaystyle\sum_{n=1}^{N_{b}}}
\frac{-\mathrm{i}c_{b,n}\xi_{x,t}\left(  \mathrm{i}\kappa_{b,n}\right)  ^{-1}%
}{k-\mathrm{i}\kappa_{b,n}}-\int_{\mathbb{R}}\frac{\xi_{x,t}\left(
\lambda\right)  ^{-1}L_{b}\left(  \lambda\right)  }{\lambda-\left(
k-\mathrm{i}0\right)  }\frac{\mathrm{d}\lambda}{2\pi\mathrm{i}};\nonumber
\end{align}
Noticing that $\mathbb{P}_{-}\left(  k-\mathrm{i}\kappa_{n}\right)
^{-1}=\left(  k-\mathrm{i}\kappa_{n}\right)  ^{-1}$ and deforming the contour
as we did above, we arrive at%
\[
\Phi_{b,x,t}\left(  k\right)  =-\int_{\Gamma}\frac{\xi_{x,t}\left(
\lambda\right)  ^{-1}L_{b}\left(  \lambda\right)  }{\lambda-\left(
k-\mathrm{i}0\right)  }\frac{\mathrm{d}\lambda}{2\pi\mathrm{i}}.
\]
Introduce%
\[
\Phi_{x,t}\left(  k\right)  =-\int_{\Gamma}\frac{\xi_{x,t}\left(
\lambda\right)  ^{-1}L\left(  \lambda\right)  }{\lambda-\left(  k-\mathrm{i}%
0\right)  }\frac{\mathrm{d}\lambda}{2\pi\mathrm{i}}.
\]
We show now that%
\begin{equation}
\left\Vert \mathbb{H}(\Phi_{x,t}-\Phi_{b,x,t})\right\Vert \rightarrow
0,\ \ \ b\rightarrow+\infty.\label{conv}%
\end{equation}
Consider%
\[
\varphi_{b}\left(  k\right)  :=\int_{\Gamma}\frac{\xi_{x,t}\left(
\lambda\right)  ^{-1}\left[  L\left(  k\right)  -L_{b}\left(  \lambda\right)
\right]  }{\lambda-\left(  k-\mathrm{i}0\right)  }\mathrm{d}\lambda.
\]
By Lemma \ref{Lemma on Cauchy transform}%
\[
\left\Vert \varphi_{b}\right\Vert _{BMO}\lesssim\left\Vert L-L_{b}\right\Vert
_{L^{\infty}\left(  \Gamma\right)  }.
\]
Since by Lemma \ref{Lemma on unif conv} $\sup_{\Gamma}\left\vert
L-L_{b}\right\vert \rightarrow0$, the Nehari theorem yields (\ref{conv}).

Since $1$ is not in $H^{2}$ we have to treat the convergence of $\mathbb{H}%
(\Phi_{x,t}-\Phi_{b,x,t})1$ separately. Observe that%
\[
\mathbb{H}(\varphi)1=\mathbb{J}\varphi
\]
if $\varphi\in H_{-}^{2}$, and thus the problem boils down to the convergence
of $\Phi_{x,t}-\Phi_{b,x,t}$ in $L^{2}$. We have
\begin{align*}
&  \left\Vert \Phi_{x,t}-\Phi_{b,x,t}\right\Vert _{2}\\
&  \leq\left\Vert \mathbb{P}_{-}\left\{  \mathrm{1}_{\left\vert \cdot
\right\vert \geq a}\xi_{x,t}{}^{-1}\left(  L-L_{b}\right)  \right\}
\right\Vert _{2}+\left\Vert \int_{S_{a}}\frac{\xi_{x,t}\left(  \lambda\right)
^{-1}\left(  L-L_{b}\right)  \left(  \lambda\right)  }{\lambda-\cdot}%
\frac{\mathrm{d}\lambda}{2\pi}\right\Vert _{2}\\
&  \leq\left\Vert L-L_{b}\right\Vert _{2}+\left\Vert \int_{S_{a}}\frac
{\xi_{x,t}\left(  \lambda\right)  ^{-1}\left(  L-L_{b}\right)  \left(
\lambda\right)  }{\lambda-\cdot}\frac{\mathrm{d}\lambda}{2\pi}\right\Vert
_{2}.
\end{align*}
By Lemma \ref{Lemma on Cauchy transform} we have
\[
\left\Vert \int_{S_{a}}\frac{\xi_{x,t}\left(  \lambda\right)  ^{-1}\left(
L-L_{b}\right)  \left(  \lambda\right)  }{\lambda-\cdot}\frac{\mathrm{d}%
\lambda}{2\pi}\right\Vert _{2}\lesssim\left\Vert \xi_{x,t}{}^{-1}\right\Vert
_{L^{\infty}\left(  S_{a}\right)  }\left\Vert L-L_{b}\right\Vert _{L^{\infty
}\left(  S_{a}\right)  },
\]
where $S_{a}=\left[  -a,-a+\mathrm{i}a,a+\mathrm{i}a,a\right]  ,$ and thus%
\[
\left\Vert \Phi_{x,t}-\Phi_{b,x,t}\right\Vert _{2}\lesssim\left\Vert
L-L_{b}\right\Vert _{2}+\left\Vert \xi_{x,t}{}^{-1}\right\Vert _{L^{\infty
}\left(  S_{a}\right)  }\left\Vert L-L_{b}\right\Vert _{L^{\infty}\left(
S_{a}\right)  }.
\]
The first term on the right hand side vanishes as $b\rightarrow\infty$ due to
Lemma \ref{Lemma on L} and the second one due to (\ref{unif conv}). Note that
$\left\Vert \xi_{x,t}{}^{-1}\right\Vert _{L^{\infty}\left(  S_{a}\right)  }$
remains bounded for each fixed $x$ and $t>0$.

We show next that%
\begin{equation}
\left\Vert \left(  I+\mathbb{H}(\Phi_{b,x,t})\right)  ^{-1}-\left(
I+\mathbb{H}(\Phi_{x,t})\right)  ^{-1}\right\Vert \rightarrow
0,\ \ b\rightarrow+\infty. \label{resolvent conv}%
\end{equation}
To this end we need to show that Lemma \ref{Lemma on unif converrgence}
applies. Noting that%
\[
\mathbb{H}\left(
{\displaystyle\sum_{n\geq1}}
\frac{-\mathrm{i}c_{b,n}\xi_{x,t}\left(  \mathrm{i}\kappa_{b,n}\right)  ^{-1}%
}{\cdot-\mathrm{i}\kappa_{b,n}}\right)  \geq0\text{ (positive operator)}%
\]
and deforming $\Gamma$ back to $\mathbb{R}$, we have
\begin{align*}
\mathbb{H}(\Phi_{b,x,t})  &  =\mathbb{H}(\varphi_{b,x,t})=\mathbb{H}\left(
{\displaystyle\sum_{n\geq1}}
\frac{-\mathrm{i}c_{b,n}\xi_{x,t}\left(  \mathrm{i}\kappa_{b,n}\right)  ^{-1}%
}{\cdot-\mathrm{i}\kappa_{b,n}}+\xi_{x,t}{}^{-1}L_{b}\right) \\
&  \geq\mathbb{H}\left(  \xi_{x,t}{}^{-1}L_{b}\right)  .
\end{align*}
It is well-know from the Marchenko theory that $\left\Vert \mathbb{H}\left(
\xi_{x,t}{}^{-1}L_{b}\right)  \right\Vert <1$ and thus%
\[
I+\mathbb{H}(\Phi_{b,x,t})\geq I+\mathbb{H}\left(  \xi_{x,t}{}^{-1}%
L_{b}\right)  \geq\alpha_{b}I
\]
with some $0<\alpha_{b}<1-\left\Vert \mathbb{H}\left(  \xi_{x,t}{}^{-1}%
L_{b}\right)  \right\Vert $. This guarantees boundedness of $\left(
I+\mathbb{H}(\Phi_{b,x,t})\right)  ^{-1}$. To justify (\ref{resolvent conv})
we need same bound for $b=\infty$ as well. To this end, let us look into
$\mathbb{H}\left(  \xi_{x,t}{}^{-1}L\right)  $. We rely on the following
important result \cite{GuillorySarason81}: If $f\in H^{\infty}\cap H^{2}$ and
$g\in H^{\infty}+C$ then $f/g\in H^{\infty}+C$ (the Guillory-Sarason
theorem).\ By the Hartman theorem, this means then that $\mathbb{H}(f/g)$ is
compact. Note that $L$ has the following structure \cite{Ryb2006}
\[
L=S/B,
\]
where $S\in H^{\infty}\cap H^{2}$ and%
\[
B\left(  k\right)  =%
{\displaystyle\prod_{n\geq1}}
\frac{k-\mathrm{i}\kappa_{n}}{k+\mathrm{i}\kappa_{n}}%
\]
is an infinite Blaschke product, which is convergent due to the Lieb-Thirring
inequality $%
{\displaystyle\sum_{n\geq1}}
\kappa_{n}\leq\frac{1}{2}\left\Vert q\right\Vert _{1}$
\cite{Hundertmarketal1998}. Recall that $H^{\infty}+C$ is an algebra (aka the
Sarason algebra) and hence since $\xi_{x,t}\in H^{\infty}+C,$ for $t>0$
\cite{GruRybSIMA15}, and $B\in H^{\infty}$ we conclude that $B\xi_{x,t}\in
H^{\infty}+C$ and by the Guillory-Sarason theorem $L/\left(  B\xi
_{x,t}\right)  {}\in H^{\infty}+C$ and hence $\mathbb{H}\left(  \xi_{x,t}%
{}^{-1}L\right)  $ is compact. This is a crucial moment in proving that
\begin{equation}
\left\Vert \mathbb{H}\left(  \xi_{x,t}{}^{-1}L\right)  \right\Vert <1.
\label{<1}%
\end{equation}
Indeed, since $\left\vert L\left(  k\right)  \right\vert <1$, $k\neq0$, Lemma
\ref{Lemma on norm} clearly applies and (\ref{<1}) immediately \ Thus $\left(
I+\mathbb{H}(\Phi_{b,x,t})\right)  ^{-1}$, $\left(  I+\mathbb{H}(\Phi
_{x,t})\right)  ^{-1}$ are bounded. Due to (\ref{conv}), Lemma
\ref{Lemma on unif converrgence} then yields (\ref{resolvent conv}). We can
now conclude that in $H^{2}$%
\begin{align*}
m_{b}\left(  \cdot,x,t\right)  -1  &  =-\left[  I+\mathbb{H}(\Phi
_{b,x,t})\right]  ^{-1}\mathbb{J}\Phi_{b,x,t}\\
&  \rightarrow-\left[  I+\mathbb{H}(\Phi_{x,t})\right]  ^{-1}\mathbb{J}%
\Phi_{x,t}=:m\left(  \cdot,x,t\right)  -1,\ \ \ b\rightarrow+\infty,
\end{align*}
where $m\left(  \cdot,x,t\right)  \in H^{2}$.

We now show that also in $H^{2}$
\[
\partial_{x}m_{b}\left(  \cdot,x,t\right)  \rightarrow\partial_{x}m\left(
\cdot,x,t\right)  ,\ \ \ b\rightarrow+\infty.
\]
Since%
\begin{align*}
&  \partial_{x}y\left(  \cdot,x,t\right)  \\
&  =\left[  I+\mathbb{H}(\Phi_{x,t})\right]  ^{-1}\partial_{x}\mathbb{H}%
\left(  \Phi_{x,t}\right)  \left[  I+\mathbb{H}(\Phi_{x,t})\right]
^{-1}\mathbb{J}\Phi_{x,t}-\left[  I+\mathbb{H}(\Phi_{x,t})\right]
^{-1}\partial_{x}\mathbb{J}\left(  \Phi_{x,t}\right)
\end{align*}
we only need to show that the differentiated Hankel operators also converge in
norm I.e. to show that%

\begin{equation}
\left\Vert \partial_{x}\mathbb{H}(\Phi_{x,t}-\Phi_{b,x,t})\right\Vert
\rightarrow0,b\rightarrow+\infty. \label{eq conv}%
\end{equation}
Since
\[
\partial_{x}\Phi_{x,t}\left(  k\right)  =\int_{\Gamma}\frac{\lambda\xi
_{x,t}\left(  \lambda\right)  ^{-1}L\left(  \lambda\right)  }{\lambda-\left(
k-\mathrm{i}0\right)  }\frac{\mathrm{d}\lambda}{\pi}%
\]
we have as before%
\[
\left\Vert \partial_{x}\mathbb{H}(\Phi_{x,t}-\Phi_{b,x,t})\right\Vert
\lesssim\left\Vert \xi_{x,t}\left(  \lambda\right)  ^{-1}\lambda\left(
L-L_{b}\right)  \left(  \lambda\right)  \right\Vert _{L^{\infty}\left(
\Gamma\right)  }%
\]
and (\ref{eq conv}) follows by Lemma \ref{Lemma on unif converrgence}.

It remains to pass to the limit in%
\begin{align*}
q_{b}\left(  x,t\right)   &  =-\partial_{x}\int_{\Gamma}\xi_{x,t}\left(
k\right)  ^{-1}L_{b}\left(  k\right)  m_{b}\left(  k,x,t\right)
\frac{\mathrm{d}k}{\pi}\\
&  =\int_{\Gamma}2\mathrm{i}k\xi_{x,t}\left(  k\right)  ^{-1}L_{b}\left(
k\right)  \frac{\mathrm{d}k}{\pi}+\int_{\Gamma}2\mathrm{i}k\xi_{x,t}\left(
k\right)  ^{-1}L_{b}\left(  k\right)  y_{b}\left(  k,x,t\right)
\frac{\mathrm{d}k}{\pi}\\
&  -\int_{\Gamma}\xi_{x,t}\left(  k\right)  ^{-1}L_{b}\left(  k\right)
\partial_{x}y_{b}\left(  k,x,t\right)  \frac{\mathrm{d}k}{\pi}\\
&  =:I_{1}+I_{2}+I_{3},
\end{align*}
where%
\[
y_{b}\left(  \cdot,x,t\right)  =-\left[  I+\mathbb{H}(\Phi_{b,x,t})\right]
^{-1}\mathbb{JP}_{-}(\Phi_{b,x,t}).
\]
By Lemma \ref{Lemma on L}%
\[
I_{1}\rightarrow\int_{\Gamma}2\mathrm{i}k\xi_{x,t}\left(  k\right)
^{-1}L\left(  k\right)  \frac{\mathrm{d}k}{\pi},\ \ \ b\rightarrow\infty.
\]
Similarly%
\[
I_{2}\rightarrow\int_{\Gamma}2\mathrm{i}k\xi_{x,t}\left(  k\right)
^{-1}L\left(  k\right)  y\left(  k,x,t\right)  \frac{\mathrm{d}k}{\pi
}\ \ b\rightarrow\infty,
\]
and%
\[
I_{3}\rightarrow-\int_{\Gamma}\xi_{x,t}\left(  k\right)  ^{-1}L\left(
k\right)  \partial_{x}y\left(  k,x,t\right)  \frac{\mathrm{d}k}{\pi
}\ \ b\rightarrow\infty.
\]
Thus (\ref{KdV solution}) is proven. Due to the Bourgain theorem, $q\left(
x,t\right)  $ is the solution to (\ref{KdV}) with initial data $q\left(
x\right)  $ for all $t>0$.
\end{proof}

A few remarks are now in order.

\begin{remark}
If the initial data $q$ are supported on $\left(  -\infty,0\right)  $ (instead
of $\left(  0,\infty\right)  $) then a better than (\ref{KdV solution})
representation for the solution $q\left(  x,t\right)  $ holds
\cite{GruRybSIMA15}, \cite{GruRybBLMS20}. Namely,%
\[
q\left(  x,t\right)  =-\partial_{x}^{2}\log\det\left(  I+\mathbb{H}\left(
\Phi_{x,t}\right)  \right)
\]
where%
\[
\Phi_{x,t}\left(  k\right)  =-\int_{\operatorname{Im}\lambda=a}\frac{\xi
_{x,t}\left(  \lambda\right)  R\left(  \lambda\right)  }{\lambda-\left(
k-\mathrm{i}0\right)  }\frac{\mathrm{d}\lambda}{2\pi\mathrm{i}}.
\]
The main difference between this symbol and the one from Theorem
\ref{main thm} is the positive power of $\xi_{x,t}\left(  \lambda\right)  $
which allows us to deform $\mathbb{R}$ to the horizontal line in
$\mathbb{C}^{+}$ where $\left\vert \xi_{x,t}\right\vert $ decays exponentially
fast. The latter in turn allows us to pass to the limit under the only
assumption that%
\[
\sup_{\left\vert I\right\vert =1}\int_{I}q_{-}<\infty,\ \ \text{where }%
q_{-}\left(  x\right)  =\max\left\{  -q\left(  x\right)  ,0\right\}  ,
\]
which does not impose much of restriction on the behavior at $-\infty$ and
includes among others completely repulsive potentials: $\left\vert R\left(
k\right)  \right\vert =1$ for a.e. $k\in\mathbb{R}$. Moreover, $q\left(
x,t\right)  $ is a meromorphic function in $x$ for each fixed $t>0$
\cite{RybCommPDEs2013}. This is a very pronounced manifestation of the
unidirectional nature of the KdV equation.
\end{remark}

\begin{remark}
We are of course not first to approach KdV with summable initial profiles $q$
but to the best of our knowledge the previous effort is concentrated on
reflectionless $q$'s (i.e. $R=0$ identically). In \cite{GesztesyDuke92} the
limit of the classical formula for pure soliton solution is studied. This
approach yields a certain subclass of reflectionless $L^{1}$ potentials.
Another approach, based on Darboux transformation is studied in
\cite{DegaspShabat94}. The complete treatment of the reflectionless case for
$L^{1}$ potentials is recently given in \cite{Hryniv21} by means of the IST.
As opposed to the previous works the latter one imposes no restriction on
norming constants $\left(  c_{n}\right)  $, save $c_{n}>0$. We emphasize that
in our restricted support case $\left(  c_{n}\right)  $ is a summable sequence
which is a general feature of initial data with support restricted to a half
line. Of course, there are no reflectionless potentials with support different
from the whole line. Finally we mention that uniform closure of reflectionless
potentials is first studied in \cite{Marchenko91} but no analysis of decay is
given therein.
\end{remark}

\begin{remark}
The main problem with general $L^{1}$ potentials, as was mentioned above, is
that the complicated small energy behavior of scattering quantities. The main
feature of our situation is that this analysis is not necessary as we, due to
analyticity of the reflection coefficient, may loosely speaking go around the
origin, the only point where bad things may happen. This provides much needed
uniform convergence. For $L^{2}$ initial data we have only uniform convergence
is lost which is a serious barrier for pushing our approach beyond $L^{1}$
initial data.
\end{remark}

\begin{remark}
We a priori know that $q\left(  x,t\right)  \in L^{2}$ but its support becomes
the whole real line once $t>0$. Moreover, $q\left(  x,t\right)  $ behaves as
the Fourier transform of $kL\left(  k\right)  \mathrm{e}^{-8\mathrm{i}k^{3}t}%
$, which is in $L^{2}$.
\end{remark}

\section{Appendix}

The following lemma is likely well-known and valid in much more generality.
Since we cannot find an exact reference we give a simple proof specific to our situation.

\begin{lemma}
\label{Lemma on Cauchy transform}Let $f\in L^{\infty}\left(  S_{a}\right)  $
where $S_{a}=\left[  -a,-a+\mathrm{i}a,a+\mathrm{i}a,a\right]  $, $a>0$, is a
polygonal contour (rectangle). If
\[
F\left(  k\right)  =\int_{S_{a}}\frac{f\left(  \lambda\right)  \mathrm{d}%
\lambda}{\lambda-k},\ \ \ k\in\mathbb{R}\text{,}%
\]
then $F\in BMO\left(  \mathbb{R}\right)  $%
\[
\left\Vert F\right\Vert _{BMO}\lesssim\left\Vert f\right\Vert _{L^{\infty
}\left(  S_{a}\right)  }.
\]
Moreover, $F\in L^{2}$ and $\left\Vert F\right\Vert _{2}\lesssim\left\Vert
f\right\Vert _{L^{\infty}\left(  S_{a}\right)  }.$
\end{lemma}

\begin{proof}
By the scaling it suffices to prove the estimate for $a=1$. Write
\[
F=F_{\mathrm{top}}+F_{+}+F_{-},
\]
where $F_{\mathrm{top}}$ is the contribution of the horizontal side $\left[
-1+\mathrm{i},1+\mathrm{i}\right]  $ and $F_{\pm}$ are the contributions of
the vertical sides $\left[  \pm1,\pm1+\mathrm{i}\right]  $. Since $\left[
-1+\mathrm{i},1+\mathrm{i}\right]  $ stays away from $\mathbb{R}$ one
immediately conclude $F_{\mathrm{top}}\in L^{\infty}\subset BMO$ and,
moreover, $\left\vert F_{\mathrm{top}}\left(  k\right)  \right\vert
\lesssim\left\Vert f\right\Vert _{\infty}/\left(  1+\left\vert k\right\vert
\right)  $ for $\left\vert k\right\vert \geq2$, so $F_{\mathrm{top}}\in L^{2}$
with
\[
\left\Vert F_{\mathrm{top}}\right\Vert _{BMO}+\left\Vert F_{\mathrm{top}%
}\right\Vert _{2}\lesssim\left\Vert f\right\Vert _{L^{\infty}\left(
S_{1}\right)  }.
\]

Next, for the right side,
\[
F_{+}\left(  k\right)  =-\int_{0}^{1}\frac{f\left(  1+\mathrm{i}y\right)
\,\mathrm{d}y}{y-\mathrm{i}\left(  1-k\right)  },
\]
which is the Cauchy transform of the density $g_{+}\left(  y\right)  =f\left(
1+\mathrm{i}y\right)  \mathrm{1}_{\left(  0,1\right)  }\left(  y\right)  $. By
the Calder\'{o}n--Zygmund theorem,
\[
\left\Vert F_{+}\right\Vert _{BMO}\lesssim\left\Vert g_{+}\right\Vert
_{\infty}\leq\left\Vert f\right\Vert _{L^{\infty}\left(  S_{1}\right)  }.
\]
Also,
\[
\left\vert F_{+}\left(  k\right)  \right\vert \leq\left\Vert f\right\Vert
_{L^{\infty}\left(  S_{1}\right)  }\int_{0}^{1}\frac{\mathrm{d}y}%
{\sqrt{\left(  1-k\right)  ^{2}+y^{2}}}=\left\Vert f\right\Vert _{L^{\infty
}\left(  S_{1}\right)  }\operatorname{arsinh}\frac{1}{\left\vert
1-k\right\vert }.
\]
The function $\operatorname{arsinh}\left(  1/\left\vert x\right\vert \right)
$ is square integrable on $\mathbb{R}$: near $x=0$ it behaves like
$\log\left(  2/\left\vert x\right\vert \right)  $, while for $\left\vert
x\right\vert \geq1$ it is $\lesssim1/\left\vert x\right\vert $. Hence
$F_{+}\in L^{2}$ and
\[
\left\Vert F_{+}\right\Vert _{2}\lesssim\left\Vert f\right\Vert _{L^{\infty
}\left(  S_{1}\right)  }.
\]

The left side $F_{-}\left(  k\right)  $ is treated in exactly the same way and
we finally conclude%
\[
\left\Vert F\right\Vert _{BMO}+\left\Vert F\right\Vert _{2}\lesssim\left\Vert
f\right\Vert _{L^{\infty}\left(  S_{1}\right)  }.
\]

\end{proof}

\begin{lemma}
\label{Lemma on unif conv}Let $q,\widetilde{q}\in L^{1}$ and $Q=\max\left(
\left\Vert q\right\Vert _{1},\left\Vert \widetilde{q}\right\Vert _{1}\right)
$. Then%
\[
\sup_{\left\vert k\right\vert \geq a}\left\vert k\left[  L\left(  k\right)
-\widetilde{L}\left(  k\right)  \right]  \right\vert \lesssim_{a,Q}\left\Vert
q-\widetilde{q}\right\Vert _{1}.
\]
for any $a~$subject to $\left(  Q/2a\right)  \exp\left(  Q/a\right)  <1.$
\end{lemma}

\begin{proof}
Recall that by (\ref{RL}) $L\left(  k\right)  =b\left(  k\right)  T\left(
k\right)  $ \ with%
\[
b\left(  k\right)  =\frac{1}{2\mathrm{i}k}\int_{\mathbb{R}}\mathrm{e}%
^{-2\mathrm{i}kx}q\left(  x\right)  m\left(  k,x\right)  \mathrm{d}%
x,\ \ \ \frac{1}{T\left(  k\right)  }=1-\frac{1}{2\mathrm{i}k}\int%
_{\mathbb{R}}q\left(  x\right)  m\left(  k,x\right)  \mathrm{d}x,
\]
where $m\left(  k,x\right)  =\mathrm{e}^{-\mathrm{i}kx}\psi(x,k).$ Moreover
\cite{Hryniv21} (see also \cite{RybBLMS2002}),%
\[
\left\Vert \widetilde{m}\left(  k,\cdot\right)  \right\Vert _{\infty
},\left\Vert m\left(  k,\cdot\right)  \right\Vert _{\infty}\leq\exp\left(
Q/\left\vert k\right\vert \right)
\]%
\[
\left\Vert m\left(  \cdot,k\right)  -\widetilde{m}\left(  \cdot,k\right)
\right\Vert _{\infty}\leq\exp\left\{  2Q/\left\vert k\right\vert \right\}
\left\Vert q-\widetilde{q}\right\Vert _{1}/\left\vert k\right\vert .
\]
It follows that%
\begin{align*}
&  2\left\vert k\right\vert \left\vert b\left(  k\right)  -\widetilde{b}%
\left(  k\right)  \right\vert \\
&  \leq\int_{\mathbb{R}}\left\vert q\left(  x\right)  -\widetilde{q}\left(
x\right)  \right\vert \left\vert m\left(  k,x\right)  \right\vert
\mathrm{d}x+\int_{\mathbb{R}}\left\vert \widetilde{q}\left(  x\right)
\right\vert \left\vert m\left(  k,x\right)  -\widetilde{m}\left(  k,x\right)
\right\vert \mathrm{d}x\\
&  \leq\left\{  \exp\left(  Q/\left\vert k\right\vert \right)  +\exp\left\{
2Q/\left\vert k\right\vert \right\}  \right\}  \left\Vert q-\widetilde{q}%
\right\Vert _{1}%
\end{align*}
and thus%
\[
\sup_{\left\vert k\right\vert \geq a}\left\vert k\left[  b\left(  k\right)
-\widetilde{b}\left(  k\right)  \right]  \right\vert \lesssim_{a,Q}\left\Vert
q-\widetilde{q}\right\Vert _{1}.
\]
Similarly%
\[
\sup_{\left\vert k\right\vert \geq a}\left\vert k\left[  1/T\left(  k\right)
-1/\widetilde{T}\left(  k\right)  \right]  \right\vert \lesssim_{a,Q}%
\left\Vert q-\widetilde{q}\right\Vert _{1}.
\]
Next,
\[
\frac{1}{\left\vert \widetilde{T}\left(  k\right)  \right\vert }\geq1-\frac
{1}{2\left\vert k\right\vert }\int_{\mathbb{R}}\left\vert \widetilde{q}\left(
x\right)  \right\vert \left\vert \widetilde{m}\left(  k,x\right)  \right\vert
\mathrm{d}x\geq1-\frac{Q}{2\left\vert k\right\vert }\exp\left(  Q/\left\vert
k\right\vert \right)  ,
\]
and hence if we choose $a>0$ so large that%
\[
\frac{Q}{2a}\exp\left(  Q/a\right)  \leq1-c<1
\]
we have $\left\vert \widetilde{T}\left(  k\right)  \right\vert \leq1/c$ for
$\left\vert k\right\vert \geq a$. Since%
\[
L-\widetilde{L}=\left(  b-\widetilde{b}\right)  \widetilde{T}-\left(
1/T-1/\widetilde{T}\right)  T\widetilde{T}b
\]
we arrive at the desirable conclusion.
\end{proof}

\begin{lemma}
\label{Lemma on norm} If $\varphi\in H^{\infty}+C$, $\left\vert \varphi
\right\vert <1$ a.e. and $\mathbb{H}\left(  \varphi\right)  $ is selfadjoint
then $\Vert\mathbb{H}(\varphi)\Vert<1$.
\end{lemma}

\begin{proof}
(By contradiction) Assume that $\Vert\mathbb{H}(\varphi)\Vert=1$. Since
$\mathbb{H}\left(  \varphi\right)  $ is selfadjoint and compact,
$\mathbb{H}(\varphi)$ must have a unimodular eigenvalue $\lambda$
($\lambda=\pm1$). For the associated normalized eigenfunction $f\in H^{2}$ we
have%
\[
\left\langle \mathbb{H}(\varphi)f,f\right\rangle =\left\langle \varphi
f,\mathbb{P}_{-}\mathbb{J}f\right\rangle =\left\langle \varphi f,\mathbb{J}%
f\right\rangle
\]
and hence by the Cauchy inequality%
\begin{align}
1  &  =\left\vert \lambda\right\vert =\left\vert \left\langle \mathbb{H}%
(\varphi)f,f\right\rangle \right\vert \leq\int_{\mathbb{R}}\left\vert
\varphi\right\vert \left\vert f\right\vert \left\vert \mathbb{J}f\right\vert
\leq\left(  \int_{\mathbb{R}}\left\vert \varphi\right\vert \left\vert
f\right\vert ^{2}\right)  ^{1/2}\ \Vert f\Vert_{2}\nonumber\label{eig}\\
&  =\left(  \int_{\mathbb{R}}\left\vert \varphi\right\vert \left\vert
f\right\vert ^{2}\right)  ^{1/2}<\Vert f\Vert_{2}=1\text{ (since }\left\vert
\varphi\right\vert <1\text{ a.e.).}\nonumber
\end{align}
Thus $\left\vert \lambda\right\vert <1$ which is a contradiction.
\end{proof}

\begin{lemma}
\label{Lemma on unif converrgence} Let $T_{n}\rightarrow T,n\rightarrow\infty
$, in norm and $\left(  I+T_{n}\right)  ^{-1},\left(  I+T\right)  ^{-1}$ are
all bounded, then $\left(  I+T_{n}\right)  ^{-1}\rightarrow\left(  I+T\right)
^{-1}$, $n\rightarrow\infty$, in norm.
\end{lemma}

\begin{lemma}
\label{Lemma on L}If $q\in L^{2}$ and $q_{b}$ is the restriction of $q$ to
$\left(  0,b\right)  $. Then $kL_{b}\left(  k\right)  $ converges to
$kL\left(  k\right)  $ in $L^{2}$.
\end{lemma}

\begin{proof}
Our proof is based upon the following representation \cite{RybNON2011}%
\begin{equation}
L=L_{b}+\frac{T_{b}^{2}L_{>b}}{1-R_{b}L_{>b}},\label{split for L}%
\end{equation}
where all quantities indexed with $b$ correspond to those for $q_{b}$ and
$L_{>b}$ is the left reflection coefficient corresponding to $q-q_{b}$. Note
that the decomposition (\ref{split for L}) goes under different names:\ layer
stripping \cite{SylWIn1999}, potential fragmentation \cite{AK01} etc.
Actually, it is totally trivial and its derivation is purely algebraic and
follows from (\ref{T}),(\ref{RL}), and (\ref{L}) and Wronskian identities. It
follows from (\ref{split for L}) that%
\begin{equation}
\left\vert L-L_{b}\right\vert =\frac{\left\vert T_{b}\right\vert
^{2}\left\vert L_{>b}\right\vert }{\left\vert 1-R_{b}L_{>b}\right\vert }%
\leq\frac{\left(  1-\left\vert R_{b}\right\vert ^{2}\right)  \left\vert
L_{>b}\right\vert }{1-\left\vert R_{b}\right\vert \left\vert L_{>b}\right\vert
}\leq2\left\vert L_{>b}\right\vert ,\label{crutial ineq}%
\end{equation}
where we have used%
\[
\left\vert T_{b}\right\vert ^{2}+\left\vert R_{b}\right\vert ^{2}=1,\left\vert
R_{b}\right\vert \leq1,\left\vert L_{>b}\right\vert \leq1.
\]
We are done if we show that $k\left\vert L_{>b}\left(  k\right)  \right\vert $
converges to zero as $b\rightarrow+\infty$. This immediately follows from the
second Zakharov-Faddeev trace formulas \cite{Zakharov71} in its final form
\cite{KillipSimon2008} that reads%
\begin{equation}
\frac{16}{3}%
{\displaystyle\sum}
\kappa_{j}^{3}+\frac{8}{\pi}\int_{\mathbb{R}}k^{2}\log\left(  1-\left\vert
L\left(  k\right)  \right\vert ^{2}\right)  ^{-1}\mathrm{d}k=\int_{\mathbb{R}%
}q^{2},\label{ZF trace 2}%
\end{equation}
for any $q\in L^{2}$ supported on $\left(  0,\infty\right)  $. Applying
(\ref{ZF trace 2}) to $q\mathrm{1}_{\left(  b,\infty\right)  }$ yields%
\[
\int_{\mathbb{R}}k^{2}\log\left(  1-\left\vert L_{>b}\left(  k\right)
\right\vert ^{2}\right)  ^{-1}\mathrm{d}k\leq\frac{\pi}{8}\int_{b}^{\infty
}q^{2}.
\]
But since $\left\vert L_{>b}\left(  k\right)  \right\vert \leq1$ one has%
\[
\log\left(  1-\left\vert L_{>b}\left(  k\right)  \right\vert ^{2}\right)
^{-1}\geq\left\vert L_{>b}\left(  k\right)  \right\vert ^{2},
\]
and we conclude that%
\[
\int_{\mathbb{R}}k^{2}\left\vert L_{>b}\left(  k\right)  \right\vert
^{2}\mathrm{d}k\leq\frac{\pi}{8}\int_{b}^{\infty}q^{2}\rightarrow
0,b\rightarrow+\infty.
\]

\end{proof}

\section{Acknowledgment}

We are thankful to Don Marshal for his drawing our attention to
\cite{GuillorySarason81}, which has played a crucial role in our approach.

\end{document}